\pdfoutput=1
\RequirePackage{ifpdf}
\ifpdf 
\documentclass[pdftex]{sigma}
\else
\documentclass{sigma}
\fi

\numberwithin{equation}{section}

\newtheorem{Theorem}{Theorem}[section]
\newtheorem{Lemma}[Theorem]{Lemma}
\newtheorem{Proposition}[Theorem]{Proposition}
 { \theoremstyle{definition}
\newtheorem{Definition}[Theorem]{Definition}
\newtheorem{Remark}[Theorem]{Remark} }

\def\a{\alpha}
\def\b{\beta}
\def\d{\delta}

\def\g{\gamma}

\begin{document}
\allowdisplaybreaks

\newcommand{\arXivNumber}{2107.?????}

\renewcommand{\PaperNumber}{069}

\FirstPageHeading

\ShortArticleName{Separation of Variables, Quasi-Trigonometric $r$-Matrices and Generalized Gaudin Models}

\ArticleName{Separation of Variables, Quasi-Trigonometric\\ $\boldsymbol{r}$-Matrices and Generalized Gaudin Models}

\Author{Taras SKRYPNYK}

\AuthorNameForHeading{T.~Skrypnyk}

\Address{Bogolyubov Institute for Theoretical Physics, 14-b Metrolohichna Str., Kyiv, 03680, Ukraine}
\Email{\href{mailto: taras.skrypnyk@unimib.it}{taras.skrypnyk@unimib.it}}

\ArticleDates{Received March 29, 2021, in final form July 07, 2021; Published online July 18, 2021}

\Abstract{We construct two new one-parametric families of separated variables for the classical Lax-integrable Hamiltonian systems governed by a one-parametric family of non-skew-symmetric, non-dynamical $\mathfrak{gl}(2)\otimes \mathfrak{gl}(2)$-valued quasi-trigonometric classical $r$-matrices. We show that for all but one classical $r$-matrices in the considered one-parametric families the corresponding curves of separation differ from the standard spectral curve of the initial Lax matrix. The proposed scheme is illustrated by an example of separation of variables for $N=2$ quasi-trigonometric Gaudin models in an external magnetic field.}


\Classification{14H70; 17B80; 37J35}

\section{Introduction}
The problem of the separation of variables is one of the most
studied, but yet the most difficult and not completely resolved
problems in the theory of integrable systems. The separated
variables are important because they give one a possibility to integrate classical equations
of motion~\cite{DKN}. They are also used to solve exactly quantum integrable
systems~\cite{SklSep}. In this context it is worth to mention, for example, the recent results in quantum separation of variables (SoV) for the higher-rank models \cite{1907.03788, 1807.11572, 1810.10996}.

 The problem of construction of separated variables can be divided into three parts:
\begin{enumerate}\itemsep=0pt
\item Construction of a set of canonical coordinates on
the phase space.
\item Proof that the constructed coordinates satisfy the so-called equations of separation.
\item Proof that the constructed system of canonical coordinates is
complete.
\end{enumerate}

The important class of integrable models are the models admitting
Lax representation. For such the models one can much advance in
the solution of the problem of separation of variables~\cite{SklSep}. In the classical case the approach of~\cite{SklSep} -- the roots of which go back to the previous papers \cite{AHH,Alber,DD,FlaMac,SklGaud,Novikov-Veselov} -- permits to construct
a set of variables that belong to a spectral curve of the Lax
matrix, playing the role of equation of separation, i.e., the
approach of \cite{SklSep} -- the so-called ``magic recipe'' -- automatically guarantees the validity of item two.

Unfortunately ``the magic recipe'' does not guarantee
that the constructed variables belonging to the spectral curve of the Lax matrix are the canonical variables indeed.
Moreover, it often happens that the number of variables produced
by the method is not equal to the dimension of the corresponding
phase space, i.e., item three is not guaranteed even if both of the
items one and two are fulfilled. Although in some cases the problem can be resolved by certain
tricks, e.g., by complementing of the set of the obtained canonical
variables by a linear integrals playing the role of the
additional momenta of separation and by finding of the
corresponding canonically conjugated variables
\cite{AHH,AHH2,SklGaud} but, unfortunately, it is not always
the case.

That is why we propose to re-consider a scheme of \cite{SklSep},
 making its starting point a construction of the complete set
of canonical coordinates (items 1, 3) and only after that start to
check item 2, i.e., to check what equations of separation the
constructed canonical coordinates satisfy. This changing of
accents permits to obtain separated variables for which the
equations of separation do not coincide with a spectral curve of
the initial Lax matrix \cite{SkrSepTrig}.

In order to find a complete set of the canonical coordinates we
use, similar to \cite{SklSep}, the method of separating
functions $B(u)$ and $A(u)$, such that zeros of the first
function generate the Poisson-commuting coordinates and the
values of the second function in these zeros generate the
canonically conjugated momenta: $B(x_i)=0$, $p_i=A(x_i)$. The most
important in this method is a Poisson algebra to be satisfied by
the functions $B(u)$ and $A(u)$ \cite{SkrDub1}:
\begin{subequations}\label{sepalg0}
\begin{gather}
 \{B(u),B(v)\}=b(u,v)B(u)-
b(v,u)B(v),
\\
 \{A(u), B(v)\}=\a(u,v)B(u)- \b(u,v)B(v),
\\
\{A(u),A(v)\}=a(u,v)B(u)- a(v,u)B(v),
\end{gather}
\end{subequations}
for some functions $a(u,v)$, $b(u,v)$, $\a(u,v)$, $\b(u,v)$ such
that the following limit holds true
\[
\lim\limits_{u\rightarrow v} (\a(u,v)B(u)- \b(u,v)B(v))=\partial_v
B(v)+\g(v) B(v).
\]
 The corresponding restrictions imposed on $B(u)$
and $A(u)$ permit their explicit construction.

In the present paper we consider the examples of
Lax-integrable models with $\mathfrak{gl}(2)$-valued Lax matrices
$L(u)=\sum_{i,j=1}^2 L_{ij}(u)X_{ij}$ possessing the tensor
Lie--Poisson bracket
\begin{equation*}
\{L(u)\otimes 1, 1\otimes L(v)\}=[r_{12}(u,v), L(u)\otimes 1] -
[r_{21}(v,u),1\otimes L(v)],
\end{equation*}
 governed by
$\mathfrak{gl}(2)\otimes \mathfrak{gl}(2)$-valued non-skew-symmetric
``quasi-trigonometric'' classical $r$-matrix
\begin{gather}
r(u,v)=\left(\frac{1}{2}\frac{(u+v)}{(v-u)}+c_1\right)X_{11}\otimes X_{11}+
\left(\frac{1}{2}\frac{(u+v)}{(v-u)}+c_2\right)X_{22}\otimes X_{22}\nonumber\\
\hphantom{r(u,v)=}{}+
\frac{v}{(v-u)}X_{12}\otimes X_{21}+ \frac{u}{(v-u)}X_{21}\otimes
X_{12},\label{qtrm0}
\end{gather}
where the parameters $c_i$ are arbitrary and $\{ X_{ij} \mid i,j= 1,2\}$ is a standard matrix basis of $\mathfrak{gl}(2)$, i.e., the matrices $ X_{ij}$ have the following matrix elements: $(X_{ij})_{\a\b}=\d_{i\a}\d_{j\b}$.

Although the mentioned above trick with the complement of the standard
Sklyanin variables is valid in the considered cases (see \cite{Kuz} for the standard trigonometric case), i.e., it helps to produce a complete set of variables
of separation, in the present paper we have succeeded in a~construction of two {\it new} one-parametric families of
separating functions for the $r$-matrix (\ref{qtrm0}) in the
subcase $c_2=-c_1$, which are {\it complete at once}
without any additional tricks.

 The first family has the following explicit form
\begin{subequations}\label{ab0}
\begin{gather}\label{btu10}
 B(u)= L_{12}(u)+ \frac{k}{u}(L_{11}(u)-L_{22}(u))-\frac{k^2}{u^2}L_{21}(u),
\\ \label{atu10}
A(u)= L_{22}(u) +\frac{k}{u}L_{21}(u)+\left(c_1+\frac{1}{2}\right)M,
\end{gather}
 \end{subequations}
where the parameter $k$ is arbitrary and
 $M$ is a special ``geometric'' integral \cite{SkrJMP2016}
possessing the following Poisson brackets with the matrix elements
of the Lax matrix:
\begin{equation*}
\{M,L_{ij}(u)\}=\operatorname{sign}(j-i)(1- \d_{ij}) L_{ij}(u).
\end{equation*}

The second family is organized in a similar way. In the case
$k\neq 0$ both the families produce the complete set of canonical
coordinates. The important property of the both sets of separated
coordinates is that they satisfy
 the spectral curve equation modified
 with the help of the geometric integral $M$. For the first family of separated variables the
 equations of separation have the following explicit form
\begin{equation}\label{sepeq0}
\det \left(L(x_j)-\left(x_jp_j-\left(c_1+\frac{1}{2}\right)M\right){\rm Id}\right)=0.
\end{equation}
 They coincides with a spectral curve of the initial Lax matrix only in the
 special case $c_1=-\frac{1}{2}$.

 The function $B(u)$ given by (\ref{btu10}) and the function $A(u)$ given by (\ref{atu10}) are obtained~-- modulo the integral~$M$~-- using the gauge transformation $K(u)$ of the initial Lax matrix (see Remarks~\ref{remark8} and~\ref{remark10}). Due to the fact that the gauge-transformed $r$-matrix $r^K(u,v)=K^{-1}(u)\otimes K^{-1}(v) r(u,v) K(u)\otimes K(v)$ does not satisfy the conditions of~\cite{SkrDub1} (see also~\cite{SkrJGP2020}), we have inserted the integral $M$ into the momenta-generating function $A(u)$ in order to achieve the needed structure~(\ref{sepalg0}) of Poisson algebra of separating functions.

Hence, the main difference of our approach with the standard one is in the form of the momenta-generating function (\ref{atu10}) and of the equations of separation (\ref{sepeq0}). It evidently leads to the SoV non-equivalent with the standard one. Our SoV is naturally characterized by another form of the Abel-type equations and of the reconstruction formulae. Besides the difference in the form of separation curves leads to the profound mathematical difference of our SoV with the standard SoV based on the usual spectral curve. In particular, the proposed SoV is not bi-Hamiltonian one (see \cite{MFP}), since~-- as the result of the $M$-shift~-- the integrals of the corresponding models do not enter into the equations of separation (\ref{sepeq0}) in the form of the Casimirs of the corresponding Poisson pencil.

In order to be more concrete we consider a class of examples of
the Lax matrices possessing the poles of the first order in the
points $\nu_1, \nu_2,\dots,\nu_N$, i.e., the Lax matrices of the
 generalized Gaudin models \cite{SkrPLA2005, SkrJGP2006, SkrJPA2007}. We devote special attention
to the case $N=2$. The corresponding Poisson-commuting Gaudin-type
Hamiltonians in an external magnetic
field are written as follows
\begin{gather*}
H_1 =
\left(\frac{1}{2}\frac{(\nu_1+\nu_2)}{(\nu_2-\nu_1)}+c_1\right)T_{11}S_{11}+\left(\frac{1}{2}\frac{(\nu_1+\nu_2)}{(\nu_2-\nu_1)}+c_2\right)T_{22}S_{22}\\
\hphantom{H_1 =}{} +
\frac{\nu_1T_{12}S_{21}}{(\nu_2-\nu_1)}
+\frac{\nu_2T_{21}S_{12}}{(\nu_2-\nu_1)}+c_1S^2_{11} +c_2S^2_{22}+2c_{11}S_{11}+2c_{22}S_{22},
\\
H_2 =
\left(\frac{1}{2}\frac{(\nu_1+\nu_2)}{(\nu_1-\nu_2)}+c_1\right)S_{11}T_{11}+\left(\frac{1}{2}\frac{(\nu_1+\nu_2)}{(\nu_1-\nu_2)}+c_2\right)S_{22}T_{22}\\
\hphantom{H_2 =}{}
+\frac{\nu_2S_{12}T_{21}}{(\nu_1-\nu_2)}
+\frac{\nu_1S_{21}T_{12}}{(\nu_1-\nu_2)}+c_1T^2_{11} +c_2T^2_{22}+2c_{11}T_{11}+2c_{22}T_{22},
\end{gather*}
where $c_{ii}$ are the components of an external magnetic field
and
\begin{equation*}
\{{S}_{ij}, {S}_{kl}\}= \d_{kj} {S}_{il}- \d_{il} {S}_{kj}, \qquad
\{{T}_{ij}, {T}_{kl}\}= \d_{kj} {T}_{il}- \d_{il} {T}_{kj}, \qquad
\{T_{ij},S_{kl}\}=0.
\end{equation*}
 For the considered $N=2$ Gaudin models we
explicitly write coordinates and momenta of separation, the
reconstruction formulae and the Abel-type equations corresponding
to the both constructed one-parametric families of separating
functions.

We remark, that our previous result \cite{SkrSepTrig} on
separation of variables in trigonometric models is recovered, in the case $c_1=c_2=0$ and for certain fixed value of the parameter $k$, from our
first family of separated coordinates after suitable
re-parametrization and gauge transformation.

The structure of the present article is the following: in Section~\ref{section2} we remind main facts about separation of variables
and about the method of separating functions, in Section~\ref{section3}
we recall the general theory of the classical $r$-matrices and
generalized Gaudin models and specialize this theory for the case
of quasi-trigonometric $r$-matrices. In Section~\ref{section4} we
construct two families of separating functions for any integrable
system governed by the quasi-trigonometric $r$-matrix. In
 Section~\ref{section5} we perform variable separation for the $N=2$
quasi-trigonometric Gaudin models. In Section~\ref{section6} we briefly conclude and discuss the on-going problems.
At last in Appendices~\ref{appendixA} and~\ref{appendixB} we give the explicit form of the reconstruction formulae for the both cases of SoV discovered in the present paper and for $N=2$ quasi-trigonometric Gaudin models.

\section{Separation of variables: general scheme}\label{section2}
\subsection{Definitions}\label{section2.1}
Let us recall the definition of Liouville integrability and
separation of variables in the theory of Hamiltonian systems
\cite{SklSep}. An integrable Hamiltonian system with $D$ degrees
of freedom is determined on a $2D$-dimensional symplectic manifold
$\mathcal{M}$ -- symplectic leaf in the Poisson manyfold
$(\mathcal{P},\{\ ,\ \})$ -- by $D$ independent first
integrals $I_j$ commuting with respect to the Poisson bracket
\[ \{I_i,I_j\} = 0, \qquad i,j= 1,\dots,D.\]
For the Hamiltonian $H$ of the system may be taken any first
integral $I_j$.

 To find separated variables means to find
(at least locally) a set of coordinates $x_i$, $p_j$, $i,j=
 1,\dots,D$ such that there exist $D$ relations --
equations of separation
\begin{equation}\label{sepeq}
 \Phi_i(x_i, p_i,I_1, \dots,I_D) = 0, \qquad i= 1,\dots,D,
 \end{equation}
where the coordinates $x_i$, $p_j$, $i,j = 1,\dots,D$ are
canonical, i.e.,
\[
\{ x_i, p_j\}=\d_{ij}, \qquad \{x_i,x_j\}=0, \qquad \{p_i,p_j\}=0, \qquad \forall\,
i,j = 1,\dots,D.
\]
The separated variables provide a way to a construction of the
action-angle coordinates from the Liouville theorem and a way
to explicit integration of the equations of motion.

Unfortunately, in the general case no algorithm is known to
construct a set of separated variables for any given integrable
system. One of the possible methods of their construction is the
so-called method of separating functions permitting one to
construct a set of canonical coordinates.

\subsection{Separating functions and canonical coordinates}\label{sepfcancoor}
 Let us remind a method of
construction of canonical coordinates using separating
functions. Generally speaking this method can be considered
independently of separation of variables. That is why in this
subsection we do not assume any special properties of the Poisson
manyfold $\mathcal{P}$ or Poisson structure $\{\ ,\ \}$. Neither
we assume integrability or existence of the Lax representation.

Let $B(u)$ and $A(u)$ be some functions of the dynamical variables
and an auxiliary parameter~$u$, which is constant with respect to
the bracket $\{\ ,\ \}$. Let the points $x_i$, $i=
 1,\dots,P$ be zeros of the function $B(u)$ and $p_i$, $i=
 1,\dots,P$ be the values of $A(u)$ in these points, i.e.,
\[ B(x_i)=0,\qquad p_i=A(x_i).\]
 We
wish to construct Poisson brackets among these new coordinates
using the Poisson brackets between $B(u)$ and $A(u)$. The
following proposition holds true~\cite{SkrDub1}.

\begin{Proposition} Let $B(x_i)=0$, $p_j=A(x_j)$. Then
\begin{alignat*}{3}
& (i)\ && \{x_i,x_j\}=\left(\frac{\{B(u),B(v)\}}{\partial_u B(u)
\partial_v B(v)}\right)\bigg|_{u=x_i,\,v=x_j}, &
\\
& (ii)\ && \{x_j,p_i\}=\left(\frac{\{A(u),B(v)\}}{
\partial_v B(v)}\right)\bigg|_{u=x_i,\,v=x_j}+ \{x_i,x_j\}(\partial_u A(u))|_{u=x_i},&
\\
& (iii) \ \ && \{p_i,p_j\}=\big(\{A(u),A(v)\}\big)\big|_{u=x_i,\,v=x_j}+\{p_i,x_j\}
(\partial_v A(v))|_{v=x_j} + \{x_i,p_j\} (\partial_u
A(u))|_{u=x_i} & \\
&&& \quad{} - \{x_i,x_j\}(\partial_u A(u)
\partial_v A(v))|_{u=x_i,\,v=x_j},&
\end{alignat*}
where $i\neq j$.
\end{Proposition}

\begin{proof}
The equalities (i)--(iii) are obtained
by the decomposition of $B(u)$, $A(u)$, $B(v)$, $A(v)$ in Taylor
power series in the neighborhood of the points $u=x_i$, $v=x_j$ in
the expressions $\{B(u),B(v)\}$, $\{A(u),B(v)\}$, $\{A(u),A(v)\}$
 and by considering the limits $u\rightarrow x_i$,
$v\rightarrow x_j$ after the calculation of the Poisson brackets.
\end{proof}

Now we are ready to formulate the following important lemma~\cite{SkrDub1}.
\begin{Lemma} \label{separation}
Let the coordinates $x_i$ and $p_j$, $i,j=1,\dots,P$ be
defined as above. Let the functions $A(u)$, $B(u)$ satisfy the
following Poisson algebra
\begin{subequations}\label{sepalg}
\begin{gather} \label{sepalg1}
\{B(u),B(v)\}=b(u,v)B(u)- b(v,u)B(v),
\\ \label{sepalg2}
\{A(u),B(v)\}=\a(u,v)B(u)- \b(u,v)B(v),
\\ \label{sepalg3}
\{A(u),A(v)\}=a(u,v)B(u)- a(v,u)B(v).
\end{gather}
\end{subequations}
Then the Poisson bracket among the functions
$x_i$ and $p_j$, $\forall\, i,j=1,\dots,P$, $i\neq j$ are
trivial
\begin{gather*}
\{x_i,x_j\}=0, \qquad \forall\, i,j=1,\dots,P,
\\
\{x_j, p_i\}=0, \qquad \text{if} \quad i\neq j,
\\
 \{p_i,p_j\}=0, \qquad \forall, i,j\in=1,\dots,P.
\end{gather*}
 If, moreover holds also the condition
\begin{equation} \label{sepalg4}
 \lim\limits_{u\rightarrow v} (\a(u,v)B(u)-
\b(u,v)B(v))=\partial_v B(v)+\g(v) B(v)
\end{equation}
 then the corresponding Poisson
brackets are canonical, i.e., $\{x_i,p_i\}=1$, $\forall\, i= 1,\dots,P$.
\end{Lemma}

\begin{Remark}\label{remark1}
 Observe, that the method of the separating
functions $A(u)$ and $B(u)$ does not, generally speaking,
guarantee that the number $P$ of the constructed canonical
variables is equal to $D$, i.e., to half of the dimension of the
generic symplectic leaf. Neither it guarantees that the
constructed canonical coordinates satisfy the equations of
separation (\ref{sepeq}) for some integrable Hamiltonian system
defined by the Poisson-commuting Hamiltonians $\{I_i,\, i=1,\dots,D\}$. Nevertheless it is often the case and it is
necessary only to find the explicit form of the corresponding
functions $ \Phi_i(x_i, p_i, I_1, \dots,I_D)$.
 In the next sections we will illustrate this by a class of new examples.
\end{Remark}

\section[Classical r-matrices and generalized Gaudin models]{Classical $\boldsymbol{r}$-matrices and generalized Gaudin models}\label{section3}
\subsection{Definition and notations}\label{section3.1}
Let $\mathfrak{g}=\mathfrak{gl}(2)$ be the Lie algebra of the general linear
group over the field of complex numbers. Let $X_{ij}$, $i,j=1,2$ be
a standard basis in $\mathfrak{gl}(2)$ with the commutation relations
\begin{equation*}
[ X_{ij}, X_{kl}]= \d_{kj} X_{il}-\d_{il}X_{kj}.
\end{equation*}

\begin{Definition}
A function of two complex
variables $r(u_1,u_2)$ with values in the tensor square of the
algebra $\mathfrak{g}=\mathfrak{gl}(2)$ is called a classical $r$-matrix if
it satisfies the following generalized classical Yang--Baxter
equation
\begin{equation}\label{GCYB}
[r_{12}(u_1,u_2), r_{13}(u_1,u_3)]= [r_{23}(u_2,u_3),
r_{12}(u_1,u_2)]-[r_{32}(u_3,u_2), r_{13}(u_1,u_3)],
\end{equation}
where
\begin{gather*}
r_{12}(u_1,u_2)\equiv
\sum\limits_{i,j,k,l=1}^{2}r^{ij,kl}(u_1, u_2)X_{ij} \otimes
X_{kl}\otimes 1, \\
r_{13}(u_1,u_3)\equiv
\sum\limits_{i,j,k,l=1}^{2}r^{ij,kl}(u_1, u_3)X_{ij} \otimes 1
\otimes X_{kl},
\end{gather*} etc. and $r^{ij,kl}(u, v)$
are matrix elements of the $r$-matrix $r(u,v)$.
\end{Definition}

\begin{Remark}\label{remark2}
The definition (\ref{GCYB}) holds true also for
other semisimple (reductive) Lie algebras~$\mathfrak{g}$. It has
appeared in the different forms in the papers
\cite{AT1, BabVia, Maillet}. In the case of
skew-symmetric $r$-matrices, i.e., when
$r_{12}(u_1,u_2)=-r_{21}(u_2,u_1)$ the generalized classical
Yang--Baxter equation reduces to the usual classical Yang--Baxter
equation \cite{BD,Skl}.
\end{Remark}

In the present paper we are interested only in the meromorphic
$r$-matrices that possess the following decomposition
\begin{equation}\label{reg}
r(u_1,u_2)=\frac{\Omega}{u_1-u_2}+ r_0(u_1,u_2),
\end{equation}
 where $r_0(u_1,u_2)$ is a holomorphic
$\mathfrak{gl}(2)\otimes \mathfrak{gl}(2)$-valued function and
$\Omega=\sum_{i,j=1}^2 X_{ij}\otimes X_{ji}$.

We will need also the following definitions:

\begin{Definition}\label{definition2}
 The classical $r$-matrix is called
$\mathfrak{g}_0\subset \mathfrak{g}$-invariant if
\begin{equation}\label{inv}
[r_{12}(u_1,u_2),X\otimes 1+ 1\otimes X]=0,
\qquad \forall\, X\in
\mathfrak{g}_0.
\end{equation}
\end{Definition}

\begin{Definition}\label{definition3}
 A $\mathfrak{gl}(2)$-valued function
$c(u)=\sum_{i,j=1}^2 c_{ij}(u)X_{ij}$ of one complex
variable is called a generalized shift element if it satisfies
the following equation
\begin{equation}\label{shrm}
[r_{12}(u_1,u_2), c(u_1)\otimes 1]- [r_{21}(u_2,u_1), 1\otimes c(u_2)]=0.
\end{equation}
\end{Definition}

\subsection{Lax algebra and generalized Gaudin models in external field}\label{section3.2}
Using the classical $r$-matrix $r(u_1,u_2)$ it is possible to
define the tensor Lie--Poisson bracket in the space of
certain $\mathfrak{gl}(2)$-valued functions of the complex parameter $u$
\cite{AT1, BabVia, Maillet}:
\begin{equation}\label{rmbr}
\{{L}(u_1)\otimes 1,1\otimes
{L}(u_2)\}=[r_{12}(u_1,u_2),{L}(u_1)\otimes 1]- [r_{21}(u_2,u_1),1
\otimes {L}(u_2)],
\end{equation}
where
\[ L(u)=\sum\limits_{i,j=1}^2 L_{ij}(u)X_{ij}.\]
The tensor bracket (\ref{rmbr}) guarantees commutativity of
spectral invariants of the Lax matrix{\samepage
\[
\big\{\operatorname{tr} L^m(u), \operatorname{tr} {L}^n(v)\big\}=0
\]
and permits one to define a completely integrable Hamiltonian
system.}

For a given classical $r$-matrix there exist different types of
the dependence of $L(u)$ on spectral parameter and dynamical
variables \cite{SkrJMP2013,SkrJMP2014, SkrJMP2016}.
In the present paper we will consider the dependence leading to
Gaudin-type models in an external magnetic field
\cite{SkrPLA2005, SkrJGP2006, SkrJPA2007}.

 In more details, let ${S}_{ij}^{(m)}$, $i,j=1,\dots,n$,
$m=1,\dots,N$ be linear coordinate functions on the dual
space to the Lie algebra $\mathfrak{gl}(2)^{\oplus N}$ with the following
Poisson brackets
\begin{equation}\label{comrelP}
\big\{{S}_{ij}^{(m)}, {S}_{kl}^{(p)}\big\}= \d^{pm}\big(\d_{kj}
{S}_{il}^{(m)}- \d_{il} {S}_{kj}^{(m)}\big).
\end{equation}
Let us fix $N$ distinct points of the complex plane $\nu_m$,
$m=1,\dots,N$. As a consequence of the generalized classical Yang--Baxter equation (\ref{GCYB}) and of the shift equation (\ref{shrm}), it is possible to introduce the following
classical Lax operator \cite{SkrPLA2005, SkrJGP2006, SkrJPA2007}:
 \begin{equation}\label{qlopg}
{L}(u)=\sum\limits_{i,j=1}^{2}{L}_{ij}(u)X_{ij}\equiv
\sum\limits_{m=1}^{N}\sum\limits_{i,j,p,q=1}^{2}
r^{ij,pq}(\nu_m,u){S}^{(m)}_{ij}X_{pq}+c(u),
\end{equation}
 where $c(u)=\sum\limits_{i,j=1}^2 c_{ij}(u)X_{ij}$ is a solution of
the shift equation (\ref{shrm}).

The corresponding Poisson-commuting Gaudin-type Hamiltonians read
as follows
\begin{gather*}
{H}_{l}=\sum\limits_{k=1,k\neq l}^N \sum\limits_{i,j, p,q=1}^{n}
r^{ij,pq}(\nu_k, \nu_l){S}_{ij}^{(k)}{S}_{pq}^{(l)}
 +\sum\limits_{i,j, p,q=1}^{n} r_0^{ij,pq}(\nu_l,
 \nu_l){S}_{ij}^{(l)}{S}_{pq}^{(l)} + \sum\limits_{i,j=1}^{n} {c}_{ij}(\nu_l)S^{(l)}_{ij},
\end{gather*}
where
 $r_0^{ij,pq}(u,v)$ are the matrix elements of the regular
part of the $r$-matrix and ${c}_{ij}(\nu_l)$ play the role of the
components of external magnetic field.

\begin{Remark}\label{remark3}
 The brackets (\ref{comrelP}) possess the
following Casimir functions
\begin{equation*}
I_{1i}=\sum\limits_{i=1}^N \big(S^{(i)}_{11}+S^{(i)}_{22}\big), \qquad
I_{2i}=\sum\limits_{i=1}^N \big(S^{(i)}_{11}S^{(i)}_{22}-
S^{(i)}_{12}S^{(i)}_{21}\big), \qquad i=1,\dots,N.
\end{equation*}
\end{Remark}

\subsection{Quasi-trigonometric case}\label{section3.3} In this subsection we will
illustrate the material of the previous subsections by the
concrete example of $\mathfrak{gl}(2)\otimes \mathfrak{gl}(2)$-valued classical
$r$-matrices and the corresponding Gaudin models.
\subsubsection[Quasi-trigonometric r-matrices]{Quasi-trigonometric $\boldsymbol{r}$-matrices}\label{section3.3.1}

 Let us consider
the following $\mathfrak{gl}(2)\otimes \mathfrak{gl}(2)$-valued function of two complex
variables
\begin{gather}
r(u,v)=\left(\frac{1}{2}\frac{(u+v)}{(v-u)}+c_1\right)X_{11}\otimes X_{11}+
\left(\frac{1}{2}\frac{(u+v)}{(v-u)}+c_2\right)X_{22}\otimes X_{22}\nonumber\\
\hphantom{r(u,v)=}{} +
\frac{v}{(v-u)}X_{12}\otimes X_{21}+ \frac{u}{(v-u)}X_{21}\otimes
X_{12}.\label{qtrm}
\end{gather}
It is possible to show \cite{SkrJGP2010} that the function
(\ref{qtrm}) satisfies generalized classical Yang--Baxter equation
(\ref{GCYB}), i.e., is a classical $r$-matrix. Using the
trigonometric parametrization $u={\rm e}^{\phi}$, $v={\rm e}^{\psi}$ it is
easy to show that the $r$-matrix~(\ref{qtrm}) satisfies the
condition~(\ref{reg}).

We will call the $r$-matrix (\ref{qtrm}) to be {\it quasi-trigonometric}.

It is easy to show that the $r$-matrix (\ref{qtrm}) possess
 a~symmetry~(\ref{inv}) with respect to the Cartan-subalgebra of the diagonal
 matrices.

It is also easy to show that the shift element for the $r$-matrix
(\ref{qtrm}) satisfying the equation~(\ref{shrm}) has the
following form
\begin{equation}\label{qtshift}
c(u)=2(c_{11}X_{11} +c_{22}X_{22}),
\end{equation}
where the overall coefficient (two) is introduced for the further
notational convenience.

\begin{Remark}\label{remark4}
 In the case $c_1=0$, $c_2=0$ the $r$-matrix~(\ref{qtrm}) coincides with skew-symmetric trigonometric
$r$-matrix. In all other cases this $r$-matrix is not
skew-symmetric.
\end{Remark}

\begin{Remark}\label{remark5}
 In the present paper we will concentrate on the
one-parametric sub-family of the $r$-matrices (\ref{qtrm})
characterized by the property $c_2=-c_1$. This family evidently
includes a~skew-symmetric point $c_1=c_2=0$ as a special partial
case.
\end{Remark}

\subsubsection{Quasi-trigonometric Lax algebra and Gaudin-type-models}\label{section3.3.2}
The Lax algebra that correspond to the
$r$-matrix (\ref{qtrm}) usually possesses \cite{SkrJMP2016} the function $M$ associated with
the Cartan-invariance of the $r$-matrix (\ref{qtrm}), that satisfies the following equality
\begin{gather}\label{m3}
\{M,L_{ij}(u)\}=\operatorname{sign}(j-i) (1-\d_{ij}) L_{ij}(u).
\end{gather}
It is easy to show \cite{SkrJMP2016} that $M$ is an integral of
motion, i.e.,
\begin{equation*}
\big\{M,\operatorname{tr} L^k(u)\big\}=0, \qquad \forall\, k \in \mathbb{N}.
\end{equation*}
As we will show below this additional integral is important for a
separation of variables.

The Lax matrix of the Gaudin-type models in an external magnetic
field (\ref{qlopg}) corresponding to the $r$-matrix (\ref{qtrm})
has the following form
\begin{gather}
L(v)=\sum\limits_{i=1}^N
\left(\left(\frac{1}{2}\frac{\nu_i+v}{\nu_i-v}+c_1\right)S^{(i)}_{11}X_{11} +
\left(\frac{1}{2}\frac{\nu_i+v}{\nu_i-v}+c_2\right)S^{(i)}_{22}X_{22}\right.\nonumber\\
\left.\hphantom{L(v)=\sum\limits_{i=1}^N}{} +
\frac{v S^{(i)}_{12}}{(\nu_i-v)}X_{21}+ \frac{\nu_i
S^{(i)}_{21}}{(\nu_i-v)}X_{12}\right)+c(v),\label{laxNqtg}
\end{gather}
where $c(v)$ is given by the formula (\ref{qtshift}).

The Poisson-commuting Gaudin-type Hamiltonians in an external
magnetic field are
\[ H_{j}=\frac{1}{2} \operatorname{res}_{v=\nu_j} \operatorname{tr} L^2(v), \qquad j=1,\dots,N.\]
They have the following explicit form
\begin{gather}
H_{j}=\sum\limits_{i=1,i \neq j}^N
\left(\left(\frac{1}{2}\frac{\nu_i+\nu_j}{\nu_i-\nu_j}+c_1\right)S^{(i)}_{11}S^{(j)}_{11}
+
\left(\frac{1}{2}\frac{\nu_i+\nu_j}{\nu_i-\nu_j}+c_2\right)S^{(i)}_{22}S^{(j)}_{22}+
\frac{\nu_j S^{(i)}_{12}S^{(j)}_{21}}{(\nu_i-\nu_j)}\right.\nonumber\\
\left.\hphantom{H_{j}=}{}
+ \frac{\nu_i
S^{(i)}_{21}S^{(j)}_{12}}{(\nu_i-\nu_j)}\right)
+c_1 (S^{(j)}_{11})^2+ c_2(S^{(j)}_{22})^2+
2c_{11}S^{(j)}_{11} +2c_{22}S^{(j)}_{22}, \qquad j=1,\dots,N.\!\!\!\label{gmahj}
\end{gather}
There additional linear integral $M$ commutes with all the
functions $H_j$ and has the form
\[
M=\frac{1}{2}\sum\limits_{j=1}^N \big(S^{(j)}_{22}-S^{(j)}_{11}\big).
\]

\begin{Remark}\label{remark6}
Observe that the Hamiltonians (\ref{gmahj}) differ from the standard trigonometric Gaudin Hamiltonians by ``dynamical'', i.e., dependent on the integral $M$ and linear Casimir functions, magnetic fields \cite{SkrNPB2019}. By the other words, all the corresponding ``quasi-trigonometric'' models are very close to the trigonometric ones, but not equivalent to them, in particular because the corresponding quasi-trigonometric $r$-matrices are not equivalent to the trigonometric one.
\end{Remark}

\section{Separation of variables: quasi-trigonometric case}\label{section4}
Let us consider the case of the Lax-integrable models governed by
quasi-trigonometric $r$-mat\-rix~(\ref{qtrm}). It occurred that in
the case $c_2=-c_1$ for such the models there exist two
one-parametric families of separating functions $A(u)$, $B(u)$. We
will consider them one by one in the next subsections.
\subsection{The first family of separated variables}\label{section4.1}

Let us consider the following linear in the matrix
elements of the Lax matrix functions
 \begin{subequations}
\begin{gather}\label{btu1}
 B(u)= L_{12}(u)+ \frac{k}{u}(L_{11}(u)-L_{22}(u))-\frac{k^2}{u^2}L_{21}(u),
\\ \label{atu1}
A(u)= L_{22}(u) +\frac{k}{u}L_{21}(u)+\left(c_1+\frac{1}{2}\right)M,
\\ \label{ctu1}
 C(u)=
 L_{21}(u).
\end{gather}
 \end{subequations}
The following proposition holds true.
\begin{Proposition}
Let the $\mathfrak{gl}(2)$-valued Lax matrix $L(u)$ satisfy tensor Poisson
bracket \eqref{rmbr} with the $r$-matrix \eqref{qtrm}. Let
$c_2=-c_1$. Then the functions $B(u)$ and $A(u)$ satisfy the
following Poisson algebra
\begin{subequations}\label{brack1}
\begin{gather}\label{brbb}
 \{B(u),B(v)\}=\frac{k}{u v}(u B(u)-v B(v)),
\\ \label{brba}
 \{A(u),B(v)\}= \frac{1}{u-v} (u B(u)-v B(v)),
\\ \label{braa}
 \{A(u),A(v)\}=0.
\end{gather}
\end{subequations}
\end{Proposition}

\begin{proof}
 The proposition is proven by the direct
calculation using the explicit form of the functions $A(u)$,
$B(v)$, the explicit form of the classical $r$-matrix, the Poisson
brackets (\ref{rmbr}) and the relation (\ref{m3}).
\end{proof}

Comparing the algebra (\ref{brack1}) with the algebra
(\ref{sepalg}) we obtain that the relations
(\ref{sepalg1})--(\ref{sepalg3}) are automatically satisfied and
the relation (\ref{sepalg4}) is satisfied after the transformation
$A(u)\rightarrow \frac{A(u)}{u}$. That is why introducing the
coordinates $x_j$, $p_j$ by the relations
\begin{equation*}
B(x_j)=0,\qquad p_j=\frac{A(x_j)}{x_j}
\end{equation*}
 we obtain that these coordinates are canonical
\begin{equation*}
\{x_i,x_j\}=0, \qquad \{p_i,p_j\}=0, \qquad \{x_j, p_i\}=\d_{ij}.
\end{equation*}

Hence we have constructed a set of canonical coordinates
associated with the quasi-tri\-go\-no\-met\-ric $r$-matrices. Now it is necessary
to show that they satisfy the equations of separation.

The following proposition holds true.
\begin{Proposition}
The functions $A(u)$, $B(u)$, $C(u)$ satisfy the following
algebraic relation
\begin{equation}\label{sepeqtrig1}
\det \left(L(u)-\left(A(u)-\left(c_1+\frac{1}{2}\right)M\right){\rm Id}\right)+B(u)C(u)=0.
\end{equation}
\end{Proposition}

\begin{proof}
The proposition is proven by the
direct calculation
using the explicit form of the functions $A(u)$, $B(v)$, $C(u)$.
\end{proof}

From the relation (\ref{sepeqtrig1}) follows that we have obtained
the equations of separation for the trigonometric models which can
be written more explicitly as follows
\begin{equation}\label{sepeqtrig1'}
\det \left(L(x_j)-\left(x_jp_j-\left(c_1+\frac{1}{2}\right)M\right){\rm Id}\right)=0.
\end{equation}

Now in order to state that the obtained coordinates are
the separated coordinates indeed, it is necessary to show that the
set of the constructed coordinates is complete on symplectic leaves
of the Lie--Poisson brackets. The following proposition holds
true.
\begin{Proposition}
Let $L(u)$ be the Lax matrix of the $N$-spin Gaudin model~\eqref{laxNqtg}. Let $k\neq 0$. Then the function $B(u)$ given
by the formula~\eqref{btu1} possesses $N$ non-constant zeros.
\end{Proposition}

\begin{proof}
The proposition is proven by a direct
calculation. Using the explicit form (\ref{btu1}) of the function
$B(u)$ and the Lax matrix (\ref{laxNqtg}) it is easy to show that
it is a rational function in $u$, nominator of which is a
polynomial of degree $N$ whose constant in $u$ term is not zero if
$k\neq 0$.

Now, using the fact that symplectic leafs of the Lie--Poisson
brackets in $\mathfrak{gl}(2)^{\oplus N}$ are $2N$-dimensional we obtain that
our construction produce the needed number of the canonical
coordinates satisfying the equations of separation
(\ref{sepeqtrig1'}), i.e., produce a complete set of the
coordinates of separation.
\end{proof}

\begin{Remark}\label{remark7}
 Observe that the constructed set of the separated coordinates is
complete even without external magnetic field, i.e., when $c_{11}=c_{22}=0$.
\end{Remark}

 \begin{Remark}\label{remark8}
 Observe, that $B(u)= L_{12}^{K}(u)$,
where $L^K(u)=K(u)^{-1} L(u) K(u)$ and
\[ K(u)=\left(
\begin{matrix}
 1 & \dfrac{k}{u} \vspace{1mm}\\
 0 & 1
\end{matrix}
\right).
\] The element $A(u)$ coincides with $L_{22}^{K}(u)$ only
in the case $c_1=-\frac{1}{2}$. Consequently the equations of
separation~(\ref{sepeqtrig1'}) do not coincide with the spectral
curve of the Lax matrix if $c_1\neq -\frac{1}{2}$.
\end{Remark}

\subsection{The second family of separated variables}\label{section4.2} Let us
consider the following linear in the matrix elements of the Lax
matrix functions
 \begin{subequations}
\begin{gather}\label{btu2}
 B(u)= L_{12}(u)+ {k}(L_{11}(u)-L_{22}(u))-{k^2}L_{21}(u),
\\ \label{atu2}
A(u)= L_{11}(u) -{k}L_{21}(u)+\left(c_1+\frac{1}{2}\right)M,
\\ \label{ctu2}
 C(u)=
 L_{21}(u).
\end{gather}
 \end{subequations}
The following proposition holds true.
\begin{Proposition}
Let the $\mathfrak{gl}(2)$-valued Lax matrix $L(u)$ satisfy tensor Poisson
bracket~\eqref{rmbr} with the $r$-matrix \eqref{qtrm}. Let
$c_2=-c_1$. Then the functions $B(u)$ and $A(u)$ satisfy the
following Poisson algebra:
\begin{subequations}\label{brack2}
\begin{gather}\label{brbb2}
 \{B(u),B(v)\}=-{k} (B(u)-B(v)),
\\ \label{brba2}
 \{A(u),B(v)\}= -\frac{u}{u-v} (B(u)-B(v)),
\\ \label{braa2}
 \{A(u),A(v)\}=0.
\end{gather}
\end{subequations}
\end{Proposition}

\begin{proof}
 The proposition is proven by the direct
calculation using the explicit form of the functions $A(u)$,
$B(v)$, the explicit form of the classical $r$-matrix, the Poisson
brackets (\ref{rmbr}) and the relation (\ref{m3}).
\end{proof}

Comparing the algebra (\ref{brack2}) with the algebra (\ref{sepalg}) we obtain that the relations
(\ref{sepalg1})--(\ref{sepalg3}) are automatically satisfied and
the relation (\ref{sepalg4}) is satisfied after the transformation
$A(u)\rightarrow -\frac{A(u)}{u}$. That is why introducing the
coordinates $x_j$, $p_j$ by the relations
\begin{equation*}
B(x_j)=0,\qquad p_j=-\frac{A(x_j)}{x_j}
\end{equation*}
 we obtain that these coordinates are canonical
\begin{equation*}
\{x_i,x_j\}=0, \qquad \{p_i,p_j\}=0, \qquad \{x_j, p_i\}=\d_{ij}.
\end{equation*}

Hence we have constructed a set of canonical coordinates
associated with the quasi-tri\-go\-no\-met\-ric $r$-matrices. Now it is necessary
to show that they satisfy some equations of separation.

The following proposition holds true.
\begin{Proposition}
The functions $A(u)$, $B(u)$, $C(u)$ satisfy the following
algebraic relation
\begin{equation}\label{sepeqtrig2}
\det \left(L(u)-\left(A(u)-\left(c_1+\frac{1}{2}\right)M\right){\rm Id}\right)+B(u)C(u)=0.
\end{equation}
\end{Proposition}

\begin{proof}
The proposition is proven by the direct calculation using the explicit form of the functions $A(u)$, $B(v)$, $C(u)$.
\end{proof}

From the relation (\ref{sepeqtrig2}) it follows that we have
obtained new equations of separation for the trigonometric models
which can be written explicitly as follows
\begin{equation}\label{sepeqtrig2'}
\det \left(L(x_j)+\left(x_jp_j+\left(c_1+\frac{1}{2}\right)M\right){\rm Id}\right)=0.
\end{equation}

Now, in order to state that the obtained coordinates are indeed
the separated coordinates it is necessary only to show that the
set of the constructed coordinates is complete on symplectic
leafs of the Lie--Poisson brackets. The following proposition
holds true.
\begin{Proposition}
Let $L(u)$ be the Lax matrix of the $N$-spin Gaudin model
\eqref{laxNqtg}. Let $k\neq 0$. Then the function $B(u)$ given
by the formula \eqref{btu1} possesses $N$ non-constant zeros.
\end{Proposition}

\begin{proof}
 The proposition is proven by a direct
calculation. Using the explicit form of the function~$B(u)$ and
the Lax matrix~(\ref{laxNqtg}) it is easy to show that it is a
rational function in $u$, nominator of which is a polynomial of
degree $N$ whose higher order term is non-zero if $k\neq 0$.

Now, using the fact that symplectic leafs of the Lie--Poisson
brackets in $\mathfrak{gl}(2)^{\oplus N}$ are $2N$-dimensional we obtain that
our construction produces the needed number of the canonical
coordinates satisfying the equations of separation~(\ref{sepeqtrig2'}), i.e., produces a complete set of the
coordinates of separation.
\end{proof}

\begin{Remark}\label{remark9}
Observe that, similar to the first set of separated variables, the constructed second set of separated variables is
complete even without external magnetic field.
\end{Remark}

\begin{Remark}\label{remark10}
Observe, that $B(u)= L_{12}^{K}(u)$,
where $L^K(u)=K^{-1} L(u) K$ and \[ K=\left(
\begin{matrix}
 1 & k \\
 0 & 1
\end{matrix}
\right).\] The element $A(u)$ coincide with $L_{11}^{K}(u)$ only
in the case $c_1=-\frac{1}{2}$. Consequently the equations of
separation~(\ref{sepeqtrig2'}) do not coincide with the spectral
curve of the Lax matrix unless $c_1=-\frac{1}{2}$.
\end{Remark}

\section[Example: N=2 quasi-trigonometric Gaudin model]{Example: $\boldsymbol{N=2}$ quasi-trigonometric Gaudin model}\label{section5}
In this section we will illustrate the results of the previous
section by the example of $N=2$ quasi-trigonometric Gaudin models.

\subsection[N=2 quasi-trigonometric Gaudin model]{$\boldsymbol{N=2}$ quasi-trigonometric Gaudin model}\label{section5.1}

Let us consider the Hamiltonians and Lax operators of $N=2$
Gaudin-type models more explicitly. Let $S^{(m)}_{ij}$, $i,j,m=1,2$ be coordinate functions on $(\mathfrak{gl}(2)\oplus \mathfrak{gl}(2))^*$
with the standard Lie--Poisson brackets~(\ref{comrelP}). Here-after
for the further notational convenience we will use the following new
notations
\[
S^{(1)}_{ij}\equiv S_{ij}, \qquad S^{(2)}_{ij}\equiv T_{ij}.
\]
The above Poisson brackets have two linear Casimir functions
\[
 I_{11}= S_{11}+S_{22}, \qquad I_{12}= T_{11}+T_{22}
\]
and two quadratic ones
\[
I_{21} = S_{11}S_{22}-S_{12}S_{21}, \qquad I_{22} = T_{11}T_{22}-T_{12}T_{21}.
\]

The Lax operator $L(v)$ of two-spin quasi-trigonometric Gaudin
model in an external magnetic field is written as follows
\begin{subequations}\label{lax2qtg}
\begin{gather}
L(v)=L^S(v)+L^T(v)+c(v),
\\
L^S(v)=\left(\frac{1}{2}\frac{\nu_1+v}{\nu_1-v}+c_1\right)S_{11}X_{11} +
\left(\frac{1}{2}\frac{\nu_1+v}{\nu_1-v}+c_2\right)S_{22}X_{22}\nonumber\\
\hphantom{L^S(v)=}{} +
\frac{vS_{12}}{(\nu_1-v)}X_{21}+
\frac{\nu_1S_{21}}{(\nu_1-v)}X_{12},
\\
L^T(v)=\left(\frac{1}{2}\frac{\nu_2+v}{\nu_2-v}+c_1\right)T_{11}X_{11} +
\left(\frac{1}{2}\frac{\nu_2+v}{\nu_2-v}+c_2\right)T_{22}X_{22}\nonumber\\
\hphantom{L^T(v)=}{} + \frac{vT_{12}}{(\nu_2-v)}X_{21}+ \frac{\nu_2T_{21}}{(\nu_2-v)}X_{12},\\
c(v)=2c_{11}X_{11} +2c_{22}X_{22}.
\end{gather}
\end{subequations}

The corresponding Gaudin-type Hamiltonians in a magnetic field have the following form
\begin{gather*}
H_1 =
\left(\frac{1}{2}\frac{(\nu_1+\nu_2)}{(\nu_2-\nu_1)}+c_1\right)T_{11}S_{11}+\left(\frac{1}{2}\frac{(\nu_1+\nu_2)}{(\nu_2-\nu_1)}+c_2\right)T_{22}S_{22}\\
\hphantom{H_1 =}{} + \frac{\nu_1T_{12}S_{21}}{(\nu_2-\nu_1)}+\frac{\nu_2T_{21}S_{12}}{(\nu_2-\nu_1)}+c_1 S^2_{11}+c_2 S^2_{22}+2c_{11}S_{11}+2c_{22}S_{22},
\\
H_2 =
\left(\frac{1}{2}\frac{(\nu_1+\nu_2)}{(\nu_1-\nu_2)}+c_1\right)S_{11}T_{11}+\left(\frac{1}{2}\frac{(\nu_1+\nu_2)}{(\nu_1-\nu_2)}+c_2\right)S_{22}T_{22}\\
\hphantom{H_2 =}{} +\frac{\nu_2S_{12}T_{21}}{(\nu_1-\nu_2)}
+\frac{\nu_1S_{21}T_{12}}{(\nu_1-\nu_2)}+c_1 T^2_{11} +c_2 T^2_{22}+2c_{11}T_{11}+2c_{22}T_{22}.
\end{gather*}
The additional linear integral $M$ associated with the
Cartan-symmetry of the $r$-matrix is
\begin{equation*}
M = -\frac{1}{2}(S_{11}-S_{22}+T_{11}-T_{22}).
\end{equation*}
It is not independent -- it is related with $H_1$, $H_2$ and linear
Casimir functions as follows
\begin{gather*}
H_1+H_2+2c_1(I_{11}+I_{12})M+c_{11}(2M-I_{11}-I_{12})-c_{22}(2M+I_{11}+I_{12})\\
\qquad{} -\left(\frac{1}{4}(c_1+c_2)\right)(2M+I_{11}+I_{12})^2=0.
\end{gather*}
 In the interesting for us case $c_2=-c_1$ we obtain that
\[
M
=\frac{1}{2}\frac{((c_{11}+c_{22})I_{11}+(c_{11}+c_{22})I_{12}-H_1-H_2)}{((I_{11}+I_{12})c_1+c_{11}-c_{22})}.
\]

\begin{Remark}\label{remark11}
 Hereafter we will consider the case of the
$r$-matrices with $c_2=-c_1$ only. For the purpose of convenience
will chose the following poles of the Lax matrix
\[ \nu_2=-\nu_1.\]
In order to simplify all the formulae we will restrict
ourselves to the consideration of the traceless case, i.e., we will hereafter
put
$I_{11}=0$, $I_{12}=0$.
\end{Remark}

\subsection{The first family of separated variables}\label{sov1}
\subsubsection{The separating functions}\label{section5.2.1}
The first family of the separated variables are given by the
separating functions (\ref{btu1})--(\ref{atu1}). Let us specify
them for the Lax matrix (\ref{lax2qtg}). We will have
\begin{gather*}
B(u)=\frac{1}{2u (u-\nu_1)(u+\nu_1)} \bigl(
\bigl(k(2c_1-1)S_{11}-2\nu_1S_{21}+k(2c_1+1)S_{22}+
k(2c_1-1)T_{11}\\
\hphantom{B(u)=}{}
+2T_{21}\nu_1+k(2c_1+1)T_{22}+4k(-c_{22}+c_{11})\bigr)u^2 +
\bigl(2k^2T_{12}+2S_{12}k^2\\
\hphantom{B(u)=}{}
+2k \nu_1S_{22}+2k \nu_1T_{11}-2k
\nu_1T_{22}-2\nu_1^2S_{21} -2\nu_1^2T_{21}-2k \nu_1S_{11}\bigr)u\\
\hphantom{B(u)=}{}
- k\bigl(\nu_1^2(2c_1+1)S_{11} +2\nu_1S_{12}k-\nu_1^2(2c_1-1)S_{22}-\nu_1^2(2c_1+1)T_{11}
-2\nu_1T_{12}k\\
\hphantom{B(u)=}{}
-\nu_1^2(2c_1-1)T_{22}-
4\nu_1^2(-c_{22}+c_{11})\bigr)\bigr),
\\
A(u) =
\frac{1}{4(u^2-\nu_1^2)}\bigl(\bigl({-}(1+2c_1)S_{11}-(1+2c_1)S_{22}
-(1+2c_1)T_{11}-(1+2c_1)T_{22}\\
\hphantom{A(u) =}{}
+8c_{22}\bigr)u^2+\bigl({-}4S_{12}k+4T_{22}\nu_1-4T_{12}k-4S_{22}\nu_1\bigr)u+
\nu_1^2(2c_1+1)S_{11}-
4kS_{12}\nu_1\\
\hphantom{A(u) =}{}
+\nu_1^2(2c_1-3)S_{22}+\nu_1^2(2c_1+1)T_{11} +4kT_{12}\nu_1+\nu_1^2(2c_1-3)T_{22}-8c_{22}\nu_1^2\bigr).
\end{gather*}

The coordinates of separation are two solutions $x_1$, $x_2$ of
the quadratic in $u$ equation
\[B(x_i)=0.\]

The canonically conjugated momenta are given by the formula
\[
p_i=\frac{A(x_i)}{x_i}, \qquad i= 1,2.
\]

\begin{Remark}\label{remark12}
Observe that in the case $k=0$ the polynomial
$B(u)$ has only one non-constant root and, that is why, can not be
used in our construction of separated variables.
\end{Remark}

\subsubsection{The equations of separation and Abel-type equations}\label{section5.2.2}

In terms of the Hamiltonians $H_1$, $H_2$, the Casimir functions
$I_{21}$, $I_{22}$ and the canonical coordinates $x_i$, $p_i$ the
curves of separation (\ref{sepeqtrig1'}) have the following form
\begin{gather}
\Phi(x_i,p_i,H_1,H_2,I_{21},I_{22})=x_i^2p_i^2+\left(\frac{H_1+H_2}{2(c_{11}-c_{22})}-2(c_{11}+c_{22})\right)x_ip_i+
\frac{x_i\nu_1I_{21}}{(\nu_1-x_i)^2}\nonumber\\
-\frac{x_i\nu_1I_{22}}{(x_i+\nu_1)^2}-\frac{(\nu_1c_{11}-x_ic_{22})H_1}{(c_{11}-
c_{22})(\nu_1-x_i)}-
\frac{(\nu_1c_{11}+x_ic_{22})H_2}{(c_{11}-c_{22})(x_i+\nu_1)}+4c_{11}c_{22}=0, \qquad
 i= 1,2.\!\!\!\label{sepeq2qtG}
\end{gather}
Using either equations of separation (\ref{sepeq2qtG}) or the
reconstruction formulas it is possible to express $H_1$, $H_2$ via
the coordinates of separation $x_i$, $p_i$, $i= 1,2$ and the
Casimir functions.

Taking into account the canonical Poisson brackets
\begin{equation*}
\{ x_i, p_j\}=\d_{ij},\qquad \{x_i,x_j\}=0, \qquad \{p_i,p_j\}=0, \qquad i,j= 1,2,
\end{equation*}
 calculating with their help the time derivatives of the coordinates of
separation
\[
\frac{{\rm d}x_i}{{\rm d}t_j}\equiv\{x_i,H_j\}
\]
and making simple transformations we obtain Abel-type equations in the differential form
\begin{gather*}
\frac{(x_1p_1(x_1-\nu_1)+2c_{11}\nu_1-2c_{22}x_1) {\rm d}x_1}{\big(4x_1p_1(c_{11}-c_{22})+H_1+H_2-4\big(c_{11}^2-c_{22}^2\big)\big)
x_1(x_1-\nu_1)} \\ \qquad {}+
\frac{(x_2p_2(x_2-\nu_1)+2c_{11}\nu_1-2c_{22}x_2) {\rm d} x_2}{\big(4x_2p_2(c_{11}-c_{22})+H_1+H_2-4\big(c_{11}^2-c_{22}^2\big)\big)x_2(x_2-\nu_1)}=- {\rm d} t_1 ,
\\
 \frac{(x_1p_1(x_1+\nu_1)-2c_{11}\nu_1-2c_{22}x_1) {\rm d}x_1}{\big(4x_1p_1(c_{11}-c_{22})+H_1+H_{22}-
 4\big(c_{11}^2-c_{22}^2\big)\big)x_1(x_1+\nu_1)} \\ \qquad {}+
 \frac{(x_2p_2(x_2+\nu_1)-2c_{11}\nu_1-2c_{22}x_2) {\rm d}x_2}{\big(4x_2p_2(c_{11}-c_{22})+H_1+H_{22}-
 4\big(c_{11}^2-c_{22}^2\big)\big)x_2(x_2+\nu_1)} =-{\rm d}t_2,
\end{gather*}
where the momenta $p_i$, $i=1,2$ are calculated using the
curves of separation (\ref{sepeq2qtG}).

\subsection{The second family of separated variables}\label{sov2}
\subsubsection{The separating functions}\label{section5.3.1}

The second family of the separated variables is obtained using the
separating functions (\ref{btu2})--(\ref{atu2}). Let us specify it
for the Lax matrix (\ref{lax2qtg}). We will have
\begin{gather*}
B(u)=\frac{1}{2\big(u^2-\nu_1^2\big)} \bigl(
k\bigl((2c_1-1)S_{11}+2kS_{12}+k(2c_1+1)S_{22}+(2c_1-1)T_{11}+2kT_{12}\\
\hphantom{B(u)=}{} +(2c_1+1)T_{22}+4(c_{11}-c_{22})\bigr)u^2 +
\nu_1\bigl(2k^2S_{12}+2kT_{11}+2kS_{22}+2T_{21}-2k^2T_{12}\\
\hphantom{B(u)=}{}
-S_{21}-2kS_{11}-2kT_{22}\bigr)u -
\nu_1^2\big(k(2c_1+1)S_{11}-2S_{21}-k(2c_1-1)S_{22}\\
\hphantom{B(u)=}{}
-k(2c_1+1)T_{11}-
2T_{21}-k(2c_1-1)T_{22}-4k(c_{11}-c_{22})\big)\bigr),
\\
A(u) =\frac{1}{4\big(\nu_1^2-u^2\big)}
\bigl(\bigl((3-2c_1)S_{11}-4S_{12}k+(-2c_1-1)S_{22}+(3-2c_1)T_{11}-4T_{12}k\\
\hphantom{A(u)=}{} +(-2c_1-1)T_{22}- 8c_{11}\bigr)u^2+
(4S_{11}\nu_1-4T_{11}\nu_1-4kS_{12}\nu_1+4kT_{12}\nu_1)u\\
\hphantom{A(u)=}{}+\nu_1^2(2c_1+1)S_{11}+\nu_1^2(2c_1+1)S_{22}
+\nu_1^2(2c_1+1)T_{11} +
\nu_1^2(2c_1+1)T_{22}+8c_{11}\nu_1^2\bigr).
\end{gather*}

The coordinates of separation are two solutions $x_1$, $x_2$ of
the quadratic in $u$ equation
\[B(x_i)=0.\]

The canonically conjugated momenta are given by the formula
\[
p_i=-\frac{A(x_i)}{x_i}, \qquad i=1,2.
\]

\begin{Remark}\label{remark13}
Observe that in the case $k=0$ the polynomial
$B(u)$ has only one root and, that is why, can not be used in our
construction of separated variables.
\end{Remark}

\subsubsection{The Abel-type equations}\label{section5.3.2}
In terms of the Hamiltonians $H_1$, $H_2$ and Casimir functions
$I_{21}$, $I_{22}$ and the canonical coordinates $x_i$, $p_i$ the
curves of the separation (\ref{sepeqtrig2'}) have the following
form
\begin{gather}
\Phi(x_i,p_i,H_1,H_2,I_{21},I_{22})=x_i^2p_i^2-\left(\frac{H_1+H_2}{2(c_{11}-c_{22})}-2(c_{11}+c_{22})\right)x_ip_i+
\frac{x_i\nu_1I_{21}}{(\nu_1-x_i)^2}\nonumber\\
-\frac{x_i\nu_1I_{22}}{(x_i+\nu_1)^2} -\frac{(\nu_1c_{11}-x_ic_{22})H_1}{(c_{11}-
c_{22})(\nu_1-x_i)}-
\frac{(\nu_1c_{11}+x_ic_{22})H_2}{(c_{11}-c_{22})(x_i+\nu_1)}+4c_{11}c_{22}=0,
\qquad i= 1,2.\!\!\!\label{sepeq2qtG2}
\end{gather}
Using either equations of separation (\ref{sepeq2qtG2}) or the
reconstruction formulae it is possible to express $H_1$, $H_2$ via
the coordinates of separation $x_i$, $p_i$, $i= 1,2$ and the
Casimir functions.

Taking into account canonical commutation relations
\begin{equation*}
\{ x_i, p_j\}=\d_{ij},\qquad \{x_i,x_j\}=0, \qquad \{p_i,p_j\}=0, \qquad i,j=
 1,2 ,
\end{equation*}
 calculating with their help the time derivatives of the coordinates of
separation
\[
\frac{{\rm d}x_i}{{\rm d}t_j}\equiv\{x_i,H_j\}, \qquad i,j = 1,2
\]
and making simple transformations we obtain the equations of
motion in the Abel-type form
\begin{gather*}
\frac{(-x_1p_1(x_1-\nu_1)+2c_{11}\nu_1-2c_{22}x_1) {\rm d}x_1}{\big(4x_1p_1(c_{11}-c_{22})-H_1-H_2+4\big(c_{11}^2-c_{22}^2\big)\big)
x_1(x_1-\nu_1)} \\
\qquad{} +
\frac{(-x_2p_2(x_2-\nu_1)+2c_{11}\nu_1-2c_{22}x_2) {\rm d}x_2}{\big(4x_2p_2(c_{11}-c_{22})-H_1-H_2+4\big(c_{11}^2-c_{22}^2\big)\big)x_2(x_2-\nu_1)}
=-{{\rm d}t_1},
\\
 \frac{(-x_1p_1(x_1+\nu_1)-2c_{11}\nu_1-2c_{22}x_1) {\rm d}x_1}{\big(4x_1p_1(c_{11}-c_{22})-H_1-H_{22}+
 4\big(c_{11}^2-c_{22}^2\big)\big)x_1(x_1+\nu_1)} \\
 \qquad{} +
 \frac{(-x_2p_2(x_2+\nu_1)-2c_{11}\nu_1-2c_{22}x_2){\rm d}x_2}{\big(4x_2p_2(c_{11}-c_{22})-H_1-H_{22}+
 4\big(c_{11}^2-c_{22}^2\big)\big)x_2(x_2+\nu_1)}
=-{{\rm d}t_2},
\end{gather*}
where the momenta $p_i$, $i=1,2$ are calculated using the
curves of separation (\ref{sepeq2qtG2}).

\begin{Remark}
Observe, that the equations of separation and Abel-type equations for our two families of separated variables do not coincide. The corresponding reconstruction formulae do not coincide either (see Appendices~\ref{appendixA} and~\ref{appendixB}).
\end{Remark}

\section{Conclusion and discussion}\label{section6}

In this paper for the classical Lax-integrable Hamiltonian systems governed by
the one-pa\-ra\-met\-ric family of non-skew-symmetric, non-dynamical
$\mathfrak{gl}(2)\otimes \mathfrak{gl}(2)$-valued quasi-trigonometric classical
$r$-matrices we have constructed two new one-pa\-ra\-met\-ric families of separated variables.
 We have shown that for all but one $r$-matrices in the
considered one-pa\-ra\-met\-ric families the corresponding curves of
separation are ``shifted'' spectral curves of the
initial Lax matrix. We have illustrated the proposed scheme by an
example of SoV for $N=2$ quasi-trigonometric Gaudin models.

It would be interesting to specify other classes of non-skew-symmetric classical $r$-matrices for which separation curves are the shifted spectral curves of the Lax matrices and to perform SoV for them. The work over this problem is now in progress.

\appendix

\section{ The reconstruction formulae for the first SoV}\label{appendixA}
Let us now reconstruct the dynamical variables $S_{ij}$, $T_{ij}$,
$i,j = 1,2$ using the variables of separation $p_i$, $x_j$
constructed in Section~\ref{sov1} and the Casimir functions. For this purpose we solve a~system of eight linear-quadratic equations
on eight variables $S_{ij}$, $T_{ij}$:
\begin{subequations}\label{rec}
\begin{gather}
(x_1+x_2)=-\bigl(2k^2T_{12}+2S_{12}k^2+2k \nu_1S_{22}+2k
\nu_1T_{11}-2k \nu_1T_{22}-2\nu_1^2S_{21}
-2\nu_1^2T_{21}\nonumber\\
\hphantom{(x_1+x_2)=}{}
-2k \nu_1S_{11}\bigr)\bigl(k(2c_1-1)S_{11}-2\nu_1S_{21}+k(2c_1+1)S_{22}+
k(2c_1-1)T_{11}\nonumber\\
\hphantom{(x_1+x_2)=}{} +2T_{21}\nu_1+k(2c_1+1)T_{22}+4k(-c_{22}+c_{11})\bigr)^{-1},
\\
x_1x_2=
k\bigl(\nu_1^2(2c_1+1)S_{11} +2\nu_1S_{12}k-\nu_1^2(2c_1-1)S_{22}-\nu_1^2(2c_1+1)T_{11}
-2\nu_1T_{12}k\nonumber\\
\hphantom{x_1x_2=}{}
-\nu_1^2(2c_1-1)T_{22}- 4\nu_1^2(-c_{22}+c_{11})\bigr)
 \bigl(k(2c_1-1)S_{11}-2\nu_1S_{21}+k(2c_1+1)S_{22}\nonumber\\
\hphantom{x_1x_2=}{}
 +
k(2c_1-1)T_{11}+2T_{21}\nu_1+k(2c_1+1)T_{22}+4k(-c_{22}+c_{11})\bigr)^{-1},
\\
4\big(x_1^2-\nu_1^2\big)x_1p_1= \bigl(\bigl(-(1+2c_1)S_{11}-(1+2c_1)S_{22}
-(1+2c_1)T_{11}-(1+2c_1)T_{22}\nonumber\\
\hphantom{4\big(x_1^2-\nu_1^2\big)x_1p_1=}{}
+8c_{22}\bigr)x_1^2+ \bigl(-4S_{12}k+4T_{22}\nu_1-4T_{12}k-4S_{22}\nu_1\bigr)x_1+
\nu_1^2(2c_1+1)S_{11}\nonumber\\
\hphantom{4\big(x_1^2-\nu_1^2\big)x_1p_1=}{}
-
4kS_{12}\nu_1+\nu_1^2(2c_1-3)S_{22}+\nu_1^2(2c_1+1)T_{11} +4kT_{12}\nu_1\nonumber\\
\hphantom{4\big(x_1^2-\nu_1^2\big)x_1p_1=}{}
+\nu_1^2(2c_1-3)T_{22}-8c_{22}\nu_1^2\bigr),
\\
4\big(x_2^2-\nu_1^2\big)x_2p_2= \bigl(\bigl(-(1+2c_1)S_{11}-(1+2c_1)S_{22}
-(1+2c_1)T_{11}-(1+2c_1)T_{22}\nonumber\\
\hphantom{4\big(x_2^2-\nu_1^2\big)x_2p_2=}{}
+8c_{22}\bigr)x_2^2+ \bigl(-4S_{12}k+4T_{22}\nu_1-4T_{12}k-4S_{22}\nu_1\bigr)x_2+
\nu_1^2(2c_1+1)S_{11}\nonumber\\
\hphantom{4\big(x_2^2-\nu_1^2\big)x_2p_2=}{}
-
4k S_{12}\nu_1+\nu_1^2(2c_1-3)S_{22}+\nu_1^2(2c_1+1)T_{11} +4kT_{12}\nu_1\nonumber\\
\hphantom{4\big(x_2^2-\nu_1^2\big)x_2p_2=}{}
+\nu_1^2(2c_1-3)T_{22}-8c_{22}\nu_1^2\bigr),
\\
I_{21}= S_{11}S_{22}- S_{12}S_{21},
\\
I_{22}= T_{11}T_{22}-T_{12}T_{21},
\\
I_{11}= S_{11}+S_{22},
\\
I_{12}= T_{11}+T_{22}.
\end{gather}
\end{subequations}

The following proposition is proven by the direct calculations.

\begin{Proposition}
 The system of equations \eqref{rec} is solved with respect of the variables $S_{ij}$, $T_{ij}$, $i,j= 1,2$ as follows
\begin{gather*}
S_{11} =
-\frac{1}{2\nu_1}\bigl(x_1^2x_2(\nu_1-x_2)^2(\nu_1^2-x_1^2)^2p_1^2-x_1x_2(x_1+x_2)(\nu_1+x_2)(\nu_1-x_2)^2(\nu_1+x_1) \\
\hphantom{S_{11} =}{}
\times (\nu_1-x_1)^2 p_2p_1 +
x_1x_2^2(\nu_1^2-x_2^2)^2(\nu_1-x_1)^2p_2^2+2x_1(x_1-x_2)(\nu_1-x_2)^2(\nu_1+x_1) \\
\hphantom{S_{11} =}{}
\times
(\nu_1-x_1)^2(c_{11}\nu_1-x_2c_{22})p_1 -2x_2(x_1-x_2)(\nu_1+x_2)(\nu_1-x_2)^2(\nu_1-x_1)^2
 \\
\hphantom{S_{11} =}{}
\times (c_{11}\nu_1
-x_1c_{22})p_2+
4x_1x_2\nu_1^3(x_1-x_2)^2I_{21} -4c_{11}\nu_1c_{22}(x_1-x_2)^2(x_2-\nu_1)^2
 \\
\hphantom{S_{11} =}{}
\times
(x_1-\nu_1)^2\bigr)
\bigl((x_1-x_2)\bigl((x_1x_2(\nu_1+x_1)(\nu_1-x_1)^2(\nu_1-x_2)p_1 -x_1x_2(\nu_1+x_2)
 \\
\hphantom{S_{11} =}{}
\times(\nu_1-x_2)^2
(\nu_1-x_1)p_2+
2(x_1-x_2)(\nu_1-x_2)(\nu_1-x_1)(c_{11}\nu_1^2+x_1c_{22}x_2)\bigr)\bigr)^{-1},
\\
S_{12} = -\frac{x_1x_2}{2k
\nu_1}\bigl(x_1^2(\nu_1-x_2)^2(\nu_1-x_1)^2(\nu_1+x_1)^2p_1^2
-2x_1x_2(\nu_1+x_2)(\nu_1-x_2)^2(\nu_1+x_1)\\
\hphantom{S_{12} =}{}
\times (\nu_1-x_1)^2p_2p_1 +
x_2^2(\nu_1-x_2)^2(\nu_1+x_2)^2(\nu_1-x_1)^2p_2^2-4x_1c_{22}(x_1-x_2)(\nu_1-x_2)^2\\
\hphantom{S_{12} =}{}
\times
(\nu_1+x_1)(\nu_1-x_1)^2p_1+
4x_2c_{22} (x_1-x_2)(x_2+\nu_1)(\nu_1-x_2)^2(\nu_1-x_1)^2p_2\\
\hphantom{S_{12} =}{}
+
4\nu_1^4(x_1-x_2)^2I_{21}+
4c_{22}^2(x_1-x_2)^2(\nu_1-x_2)^2(\nu_1-x_1)^2\bigr)
 \bigl((\nu_1-x_2)(\nu_1-x_1)\\
\hphantom{S_{12} =}{}
\times
 (x_1-x_2)\bigl(x_1x_2\big(\nu_1^2-x_1^2\big)p_1-x_1x_2\big(a^2_1-x^2_2\big)p_2+2(x_1-x_2)\\
\hphantom{S_{12} =}{}
\times
\big(c_{11}\nu_1^2+x_1c_{22}x_2\big)\bigr)\bigr)^{-1},
\\
S_{21} =
\frac{k}{2 \nu_1 x_1x_2}\bigl(x_1^2x_2^2(\nu_1-x_2)^2(\nu_1-x_1)^2(\nu_1+x_1)^2p_1^2-2x_1^2x_2^2(\nu_1+x_2)(\nu_1-x_2)^2\\
\hphantom{S_{21} =}{}
\times (\nu_1+x_1)(\nu_1-x_1)^2 p_2p_1
+ x_1^2x_2^2(\nu_1-x_2)^2(\nu_1+x_2)^2(\nu_1-x_1)^2p_2^2+4c_{11}
\nu_1x_1x_2\\
\hphantom{S_{21} =}{}
\times
(x_1-x_2)(\nu_1-x_2)^2(\nu_1+x_1)(\nu_1-x_1)^2p_1 -4c_{11}\nu_1x_1x_2(x_1-x_2)(\nu_1+x_2)\\
\hphantom{S_{21} =}{}
\times
(\nu_1-x_2)^2(\nu_1-x_1)^2p_2+
4x_1^2\nu_1^2x_2^2(x_1-x_2)^2I_{21} +4\nu_1^2c_{11}^2(x_1-x_2)^2(\nu_1-x_2)^2\\
\hphantom{S_{21} =}{}
\times
(\nu_1-x_1)^2\bigr)
\bigl(x_1x_2(x_1-x_2)(\nu_1-x_2)(\nu_1+x_1)(\nu_1-x_1)^2p_1 -x_1x_2(x_1-x_2)\\
\hphantom{S_{21} =}{}
\times
(\nu_1+x_2)(\nu_1-x_2)^2(\nu_1-x_1)p_2+
2(x_1-x_2)^2(\nu_1-x_2)(\nu_1-x_1)\\
\hphantom{S_{21} =}{}
\times
(c_{11}\nu_1^2+x_1c_{22}x_2)\bigr)^{-1},
\\
S_{22}=-S_{11},
\\
T_{11}=-\frac{1}{2\nu_1}
\bigl(-x_1^2x_2(\nu_1+x_2)^2(\nu_1-x_1)^2(\nu_1+x_1)^2p_1^2+x_1x_2(x_1+x_2)(\nu_1-x_2)
(\nu_1+x_2)^2\\
\hphantom{T_{11}=}{}
\times
(\nu_1-x_1)(\nu_1+x_1)^2 p_2p_1-x_1x_2^2(\nu_1-x_2)^2(\nu_1+x_2)^2(\nu_1+x_1)^2p_2^2\\
\hphantom{T_{11}=}{}
-2x_1(-x_2+x_1)
(\nu_1+x_2)^2(\nu_1-x_1)(\nu_1+x_1)^2(c_{11}\nu_1+x_2c_{22})p_1 +
2x_2(-x_2+x_1)\\
\hphantom{T_{11}=}{}
\times
(\nu_1-x_2)(\nu_1+x_2)^2(\nu_1+x_1)^2(c_{11}\nu_1+x_1c_{22})p_2+4x_1x_2\nu_1^3(-x_2+x_1)^2I_{22} \\
\hphantom{T_{11}=}{}
-4c_{11}\nu_1c_{22}(-x_2+x_1)^2(\nu_1+x_2)^2(\nu_1+x_1)^2\bigr)
\bigl(\big(x_1x_2(\nu_1-x_1)(\nu_1+x_1)p_1\\
\hphantom{T_{11}=}{}
-x_1x_2(\nu_1-x_2)(\nu_1+x_2)p_2 +
2(-x_2+x_1)\big(c_{11}\nu_1^2x_1c_{22}x_2\big)\big)(-x_2+x_1)(\nu_1+x_2)\\
\hphantom{T_{11}=}{}
\times
(\nu_1+x_1)\bigr)^{-1},
\\
T_{12} = \frac{x_1x_2}{2k
\nu_1}\bigl(x_1^2(\nu_1+x_2)^2(\nu_1-x_1)^2(\nu_1+x_1)^2p_1^2-2x_1x_2(\nu_1-x_2)(\nu_1+x_2)^2(\nu_1-x_1)\\
\hphantom{T_{12} =}{}
\times
(\nu_1+x_1)^2p_2p_1 + x_2^2(\nu_1-x_2)^2(\nu_1+x_2)^2(\nu_1+x_1)^2p_2^2+
4x_1c_{22}(-x_2+x_1)\\
\hphantom{T_{12} =}{}
\times
(\nu_1+x_2)^2(\nu_1-x_1)(\nu_1+x_1)^2p_1-
4x_2c_{22} (-x_2+x_1)(\nu_1-x_2)(\nu_1+x_2)^2(\nu_1+x_1)^2\\
\hphantom{T_{12} =}{}
\times
p_2+4\nu_1^4(-x_2+x_1)^2I_{22}+
4c_{22}^2(-x_2+x_1)^2(\nu_1+x_2)^2(\nu_1+x_1)^2\bigr)
\bigl((x_1x_2(\nu_1^2-x_1^2)p_1\\
\hphantom{T_{12} =}{}
\times
-x_1x_2(\nu_1^2-x_2^2)p_2+2(x_1-x_2)(c_{11}\nu_1^2+
x_1c_{22}x_2))(x_1-x_2)(\nu_1+x_2)(\nu_1+x_1)\bigr)^{-1},
\\
T_{21} =
-\frac{k}{2\nu_1x_1x_2}\bigl(x_1^2x_2^2(\nu_1+x_2)^2(\nu_1-x_1)^2(\nu_1+x_1)^2p_1^2
-2x_1^2x_2^2(\nu_1-x_2)(\nu_1+x_2)^2\\
\hphantom{T_{21} =}{}
\times
(\nu_1-x_1)(\nu_1+x_1)^2 p_2p_1+
x_1^2x_2^2(\nu_1-x_2)^2(\nu_1+x_2)^2(\nu_1+x_1)^2p_2^2
+4c_{11}\nu_1x_1x_2\\
\hphantom{T_{21} =}{}
\times
(-x_2+x_1)
(\nu_1+x_2)^2(\nu_1-x_1)(\nu_1+x_1)^2p_1 -4c_{11}\nu_1x_1x_2(-x_2+x_1)(\nu_1-x_2)\\
\hphantom{T_{21} =}{}
\times
(\nu_1+x_2)^2(\nu_1+x_1)^2p_2+
4x_1^2\nu_1^2x_2^2(-x_2+x_1)^2I_{22} +4\nu_1^2c_{11}^2(-x_2+x_1)^2(\nu_1+x_2)^2\\
\hphantom{T_{21} =}{}
\times
(\nu_1+x_1)^2\bigr)
\bigl(\big(x_1x_2(x_1-x_2)(\nu_1+x_2)(\nu_1-x_1)(\nu_1+x_1)^2p_1-x_1x_2(x_1-x_2)\\
\hphantom{T_{21} =}{}
\times
(\nu_1-x_2)
(\nu_1+x_2)^2(\nu_1+x_1)p_2+
2(x_1-x_2)^2(\nu_1+x_2)(\nu_1+x_1)\\
\hphantom{T_{21} =}{}
\times
\big(c_{11}\nu_1^2+x_1c_{22}x_2\big)\big)\bigr)^{-1},
\\
T_{22}=-T_{11},
\end{gather*}
where for the purpose of simplicity we have
put $I_{11}=0$, $I_{12}=0$.
\end{Proposition}

\section{The reconstruction formulae for the second SoV}\label{appendixB}
Let us now reconstruct the dynamical variables $S_{ij}$, $T_{ij}$,
$i,j= 1,2$ using the variables of separation $p_i$,
$x_j$ constructed in Section~\ref{sov2} and the Casimir functions. For this purpose we solve a~system of eight linear-quadratic
equations in eight variables $S_{ij}$, $T_{ij}$:
\begin{subequations}\label{rec2}
\begin{gather}
(x_1+x_2)=-\frac{\nu_1}{k}(2k^2S_{12}+2kT_{11}+2kS_{22}+2T_{21}-2k^2T_{12}-S_{21}-2kS_{11}-2kT_{22})\nonumber\\
\hphantom{(x_1+x_2)=}{}
\times
\bigl((2c_1-1)S_{11}+2kS_{12}+k(2c_1+1)S_{22}+(2c_1-1)T_{11}+2kT_{12}\nonumber\\
\hphantom{(x_1+x_2)=}{}
+(2c_1+1)T_{22}+4(c_{11}-c_{22})\bigr)^{-1},
\\
x_1x_2=
\frac{\nu_1^2}{k}(k(2c_1+1)S_{11}-2S_{21}-k(2c_1-1)S_{22}-k(2c_1+1)T_{11}-
2T_{21}-k(2c_1-1)T_{22}\nonumber\\
\hphantom{x_1x_2=}{}
-4k(c_{11}-c_{22})
\bigl((2c_1-1)S_{11}+2kS_{12}+k(2c_1+1)S_{22}+(2c_1-1)T_{11}+2kT_{12}\nonumber\\
\hphantom{x_1x_2=}{}
+(2c_1+1)T_{22}+4(c_{11}-c_{22})\bigr)^{-1},
\\
-4\big(x_1^2-\nu_1^2\big)x_1p_1=
\bigl(((3-2c_1)S_{11}-4S_{12}k+(-2c_1-1)S_{22}+(3-2c_1)T_{11}-4T_{12}k\nonumber\\
\hphantom{-4\big(x_1^2-\nu_1^2\big)x_1p_1=}{}
+(-2c_1-1)T_{22}- 8c_{11}) x_1^2+
(4S_{11}\nu_1-4T_{11}\nu_1-4kS_{12}\nu_1\nonumber\\
\hphantom{-4\big(x_1^2-\nu_1^2\big)x_1p_1=}{}
+4kT_{12}\nu_1)x_1+\nu_1^2(2c_1+1)S_{11}
+\nu_1^2(2c_1+1)S_{22}
+\nu_1^2(2c_1+1)T_{11}
\nonumber\\
\hphantom{-4\big(x_1^2-\nu_1^2\big)x_1p_1=}{}
 +
\nu_1^2(2c_1+1)T_{22}+8c_{11}\nu_1^2\bigr),
\\
-4\big(x_2^2-\nu_1^2\big)x_2p_2=
\bigl(((3-2c_1)S_{11}-4S_{12}k+(-2c_1-1)S_{22}+(3-2c_1)T_{11}-4T_{12}k\nonumber\\
\hphantom{-4\big(x_2^2-\nu_1^2\big)x_2p_2=}{}
+(-2c_1-1)T_{22}- 8c_{11}) x_2^2 +
(4S_{11}\nu_1-4T_{11}\nu_1-4kS_{12}\nu_1
\nonumber\\
\hphantom{-4\big(x_2^2-\nu_1^2\big)x_2p_2=}{}
+4kT_{12}\nu_1)x_2+\nu_1^2(2c_1+1)S_{11}+\nu_1^2(2c_1+1)S_{22}
+\nu_1^2(2c_1+1)T_{11}
\nonumber\\
\hphantom{-4\big(x_2^2-\nu_1^2\big)x_2p_2=}{}
 +
\nu_1^2(2c_1+1)T_{22}+8c_{11}\nu_1^2\bigr),
\\
I_{21}= S_{11}S_{22}- S_{12}S_{21},
\\
I_{22}= T_{11}T_{22}-T_{12}T_{21},
\\
I_{11}= S_{11}+S_{22},
\\
I_{12}= T_{11}+T_{22}.
\end{gather}
\end{subequations}

The following proposition is proven by the direct calculations.
\begin{Proposition}
The system of equations \eqref{rec2} is solved with respect of the variables $S_{ij}$, $T_{ij}$, $i,j= 1,2$ as follows
\begin{gather*}
S_{11}
=-\frac{1}{2\nu_1}\bigl(x_1^2x_2(-x_2+\nu_1)^2\big(\nu_1^2-x_1^2\big)^2p_1^2-x_2x_1(x_1+x_2)(x_2+\nu_1)
(\nu_1-x_2)^2\\
\hphantom{S_{11}=}{}
\times\big(\nu_1^2-x_1^2\big)^2p_2p_1+ x_2^2x_1(-x_2+\nu_1)^2(x_2+\nu_1)^2(\nu_1-x_1)^2p_2^2-2x_1(-x_2+x_1)\\
\hphantom{S_{11}=}{}
\times
(-x_2+\nu_1)^2(x_1+\nu_1)(-x_1+\nu_1)^2(c_{11}\nu_1-c_{22}x_2)) p_1
+2x_2(x_1-x_2)(x_2+\nu_1)\\
\hphantom{S_{11}=}{}
\times
(\nu_1-x_2)^2(\nu_1-x_1)^2(c_{11}\nu_1-c_{22}x_1)p_2+4x_1\nu_1^3x_2(x_1-x_2)^2I_{21} -
4c_{11}\nu_1c_{22}\\
\hphantom{S_{11}=}{}
\times
(x_1-x_2)^2(\nu_1-x_2)^2(\nu_1-x_1)^2\bigr)
\bigl(\big({-}x_1x_2(x_1+\nu_1)(\nu_1-x_1)p_1+x_1x_2(x_2+\nu_1)\\
\hphantom{S_{11}=}{}
\times
(\nu_1-x_2)p_2 +2(x_1-x_2)\big(x_1c_{22}x_2+c_{11}\nu_1^2\big)\big)
(x_1-x_2)(\nu_1-x_2)(\nu_1-x_1)\bigr)^{-1},
\\
S_{12}
=-\frac{1}{2kx_1x_2}\bigl(x_2^2x_1^2(\nu_1-x_2)^2(\nu_1-x_1)^2(x_1+\nu_1)^2p_1^2-2x_2^2x_1^2
(x_2+\nu_1)(\nu_1-x_2)^2\\
\hphantom{S_{12}=}{}
\times
(x_1+\nu_1)(\nu_1-x_1)^2 p_2p_1+
x_2^2x_1^2(\nu_1-x_2)^2(x_2+\nu_1)^2(\nu_1-x_1)^2p_2^2-4\nu_1x_2c_{11}x_1\\
\hphantom{S_{12}=}{}
\times
(x_1-x_2)(\nu_1-x_2)^2(x_1+\nu_1)(\nu_1-x_1)^2p_1 +4\nu_1x_2c_{11}x_1(x_1-x_2)(x_2+\nu_1)\\
\hphantom{S_{12}=}{}
\times
(\nu_1-x_2)^2(\nu_1-x_1)^2p_2+
4x_2^2\nu_1^2x_1^2(x_1-x_2)^2I_{21} +4\nu_1^2c_{11}^2(x_1-x_2)^2(\nu_1-x_2)^2\\
\hphantom{S_{12}=}{}
\times
(\nu_1-x_1)^2\bigr)
\bigl(\big({-}x_2x_1(x_1-x_2)(\nu_1-x_2)(x_1+\nu_1)(\nu_1-x_1)^2p_1
 +x_2x_1(x_1-x_2)\\
\hphantom{S_{12}=}{}
\times
 (x_2+\nu_1)(\nu_1-x_2)^2(\nu_1-x_1)p_2+2(x_1-x_2)^2(\nu_1-x_2)(\nu_1-x_1)\\
\hphantom{S_{12}=}{}
\times
 \big(x_1c_{22}x_2+
c_{11}\nu_1^2\big)\big)\bigr)^{-1},
\\
S_{21} =
\frac{x_2x_1k}{2\nu_1^2}\bigl(x_1^2(\nu_1-x_2)^2(\nu_1-x_1)^2(x_1+\nu_1)^2p_1^2-2x_1x_2(x_2+\nu_1)(\nu_1-x_2)^2(x_1+\nu_1)\\
\hphantom{S_{21} =}{}
\times
(\nu_1-x_1)^2p_2p_1 +x_2^2(\nu_1-x_2)^2(x_2+\nu_1)^2(\nu_1-x_1)^2p_2^2+
4c_{22}x_1(x_1-x_2)(\nu_1-x_2)^2\\
\hphantom{S_{21} =}{}
\times
(x_1+\nu_1)(\nu_1-x_1)^2p_1-
4c_{22}x_2 (x_1-x_2)(x_2+\nu_1)(\nu_1-x_2)^2(\nu_1-x_1)^2p_2\\
\hphantom{S_{21} =}{}
+4\nu_1^4(x_1-x_2)^2I_{21}+4c_{22}^2
(x_1-x_2)^2(\nu_1-x_2)^2(\nu_1-x_1)^2\bigr)
\bigl(\big({-}x_1x_2(\nu_1-x_1)\\
\hphantom{S_{21} =}{}
\times
(x_1+\nu_1)(x_1-x_2)p_1+x_1x_2(\nu_1-x_2)(x_2+\nu_1)(x_1-x_2)p_2 +
2(x_1-x_2)^2\\
\hphantom{S_{21} =}{}
\times
\big(x_1c_{22}x_2+c_{11}\nu_1^2\big)\big)(\nu_1-x_2)(\nu_1-x_1)\bigr)^{-1},
\\
S_{22}=-S_{11},
\\
T_{11}=
-\frac{1}{2\nu_1}\bigl(-x_1^2x_2(x_2+\nu_1)^2\big(\nu_1^2-x_1^2\big)^2p_1^2+
x_1x_2(x_1+x_2)(\nu_1-x_2)(x_2+\nu_1)^2(\nu_1-x_1)\\
\hphantom{T_{11}=}{}
\times
(x_1+\nu_1)^2p_2p_1-
x_2^2x_1(\nu_1-x_2)^2(x_2+\nu_1)^2(x_1+\nu_1)^2p_2^2+2x_1(x_1-x_2)
\\
\hphantom{T_{11}=}{}
\times
(x_2+\nu_1)^2(\nu_1-x_1)(x_1+\nu_1)^2(c_{11}\nu_1+c_{22}x_2))p_1 -2x_2(x_1-x_2)(\nu_1-x_2)
\\
\hphantom{T_{11}=}{}
\times
(x_2+\nu_1)^2(x_1+\nu_1)^2(c_{11}\nu_1+c_{22}x_1)p_2+4x_1\nu_1^3x_2(x_1-x_2)^2I_{22} -
4c_{11}\nu_1c_{22}\\
\hphantom{T_{11}=}{}
\times
(x_1-x_2)^2(x_2+\nu_1)^2(x_1+\nu_1)^2\bigr)
\bigl((x_1-x_2)\big({-}x_1x_2(\nu_1-x_1)(x_1+\nu_1)^2(x_2+\nu_1)p_1 \\
\hphantom{T_{11}=}{}
+
x_1x_2(\nu_1-x_2)(x_2+\nu_1)^2(x_1+\nu_1)p_2+2(x_1-x_2)(x_2+\nu_1)(x_1+\nu_1)\\
\hphantom{T_{11}=}{}
\times
\big(x_1c_{22}x_2+c_{11}\nu_1^2\big)\big)\bigr)^{-1},
\\
T_{12}=
-\frac{x_1x_2k}{2}\bigl(x_1^2x_2^2(x_2+\nu_1)^2\big(\nu_1^2-x_1^2\big)^2p_1^2
-2x_1^2x_2^2(\nu_1-x_2)(x_2+\nu_1)^2(\nu_1-x_1)\\
\hphantom{T_{12}=}{}
\times
(x_1+\nu_1)^2p_2p_1 +
x_1^2x_2^2(\nu_1-x_2)^2(x_2+\nu_1)^2(x_1+\nu_1)^2p_2^2-4\nu_1x_2c_{11}x_1
(x_1-x_2)\\
\hphantom{T_{12}=}{}
\times
(x_2+\nu_1)^2(\nu_1-x_1)(x_1+\nu_1)^2p_1 +4\nu_1x_2c_{11}x_1(x_1-x_2)(\nu_1-x_2)(x_2+\nu_1)^2\\
\hphantom{T_{12}=}{}
\times
(x_1+\nu_1)^2p_2+
4\nu_1^2x_2^2x_1^2(x_1-x_2)^2I_{22} +4\nu_1^2c_{11}^2(x_1-x_2)^2(x_2+\nu_1)^2(x_1+\nu_1)^2\bigr)\\
\hphantom{T_{12}=}{}
\times
\bigl(\big({-}x_1x_2(x_1-x_2)(x_2+\nu_1)(\nu_1-x_1)(x_1+\nu_1)^2p_1 +x_1x_2(x_1-x_2)(\nu_1-x_2)\\
\hphantom{T_{12}=}{}
\times
(x_2+\nu_1)^2(x_1+\nu_1)p_2+
2(x_1-x_2)^2(x_2+\nu_1)(x_1+\nu_1)\big(x_1c_{22}x_2+c_{11}\nu_1^2\big)\big)\bigr)^{-1},
\\
T_{21}= \frac{kx_1x_2}{2
\nu_1^2}\bigl(x_1^2(x_2+\nu_1)^2\big(\nu_1^2-x_1^2\big)^2p_1^2
-2x_1x_2(\nu_1-x_2)(x_2+\nu_1)^2(\nu_1-x_1)\\
\hphantom{T_{21}=}{}
\times
(x_1+\nu_1)^2p_2p_1+
x_2^2(\nu_1-x_2)^2(x_2+\nu_1)^2(x_1+\nu_1)^2p_2^2-
4c_{22}x_1(x_1-x_2)\\
\hphantom{T_{21}=}{}
\times
(x_2+\nu_1)^2(\nu_1-x_1)(x_1+\nu_1)^2p_1+
4c_{22}x_2 (x_1-x_2)
(\nu_1-x_2)(x_2+\nu_1)^2\\
\hphantom{T_{21}=}{}
\times
(x_1+\nu_1)^2p_2+4\nu_1^4(x_1-x_2)^2I_{22}+
4c_{22}^2(x_1-x_2)^2
(x_2+\nu_1)^2(x_1+\nu_1)^2\bigr)\\
\hphantom{T_{21}=}{}
\times
 \bigl((-x_1x_2(\nu_1-x_1)
(x_1+\nu_1)(x_1-x_2)p_1+x_1x_2(\nu_1-x_2)(x_2+\nu_1)(x_1-x_2)p_2 \\
\hphantom{T_{21}=}{}
+
2(x_1-x_2)^2(x_1c_{22}x_2+c_{11}\nu_1^2))(x_2+\nu_1)(x_1+\nu_1)\bigr)^{-1},
\\
T_{22}=-T_{11},
\end{gather*}
where for the purpose of simplicity we have put
 $I_{11}=0$, $I_{12}=0$.
\end{Proposition}

\subsection*{Acknowledgements}

The author is grateful to the late Boris Dubrovin for the discussions. The work over this paper was partially supported by the Division of Physics and Astronomy of NAS of Ukraine (Project No.~0117U000240).

\pdfbookmark[1]{References}{ref}
\LastPageEnding

\end{document}